\documentclass[11pt,a4paper]{article}
\usepackage[margin=2.15cm]{geometry}
\usepackage{amsmath,amssymb,amsthm,mathtools,bm,bbold}
\usepackage{booktabs,array,microtype}
\usepackage{authblk}
\usepackage[T1]{fontenc}
\usepackage{lmodern}
\usepackage{xcolor}
\usepackage[colorlinks=true,linkcolor=blue!45!black,citecolor=blue!45!black,urlcolor=blue!45!black]{hyperref}
\usepackage[nameinlink,capitalise]{cleveref}

\newtheorem{theorem}{Theorem}
\newtheorem{proposition}{Proposition}

\newcommand{\tr}{\operatorname{tr}}

\newcommand{\ii}{\mathrm{i}}
\newcommand{\cH}{\mathcal H}

\title{Scalene Yang--Baxter triples as a source of hidden symmetries beyond the ordinary Yang--Baxter equation}
\author[1]{Pramod Padmanabhan\thanks{E-mail: \texttt{pramod23phys@gmail.com}}}
\author[1]{Somnath Maity\thanks{E-mail: \texttt{somnathmaity126@gmail.com}}}
\author[2]{Vladimir E. Korepin\thanks{E-mail: \texttt{vladimir.korepin@stonybrook.edu}}}
\affil[1]{Department of Physics, School of Basic Sciences, Indian Institute of Technology Bhubaneswar, Argul, Odisha 752050, India}
\affil[2]{C. N. Yang Institute for Theoretical Physics, Stony Brook University, Stony Brook, New York 11794-3840, USA}
\date{\today}

\begin{document}
\maketitle

\begin{abstract}
We study a nearest-neighbor non-Hermitian spin chain obtained from one member of an exact non-braided scalene Yang--Baxter triple.  Its local Hamiltonian density violates both the difference-form Reshetikhin condition and its general non-difference counterpart, obstructing its realization by a differentiable homogeneous regular solution of the ordinary Yang--Baxter equation.  The transfer matrices constructed from the regular member do not commute among themselves at distinct spectral parameters.  Nevertheless, the scalene Yang--Baxter relation implies cross-commutativity with another transfer matrix constructed from the third member of the scalene triple.  We evaluate the latter for arbitrary chain length and show that, on even periodic chains, it is a finite generating function of a non-obvious staggered nilpotent symmetry.  The resulting conserved hierarchy belongs entirely to the algebra generated by this single symmetry and hence does not constitute an extensive family of algebraically independent charges.  Nevertheless, this example demonstrates that scalene Yang--Baxter triples can act as an algebraic symmetry-discovery mechanism beyond the ordinary self-commuting transfer-matrix framework.
\end{abstract}

\section{Introduction}
The quantum Yang--Baxter equation (YBE) is the standard algebraic source of commuting transfer matrices and local conserved charges in one-dimensional quantum systems \cite{Faddeev1996,KorepinBogoliubovIzergin1993,Baxter1982}.  For a regular $R$-matrix, $R(u,u)=P$, logarithmic derivatives of its transfer matrix generate a mutually commuting hierarchy whose first nontrivial member is a nearest-neighbour Hamiltonian.  At the Hamiltonian level, the resulting densities obey strong local compatibility conditions usually referred to as Reshetikhin conditions.  These conditions underlie modern classifications of regular $4\times4$ solutions, both in difference and non-difference form \cite{deLeeuwPribytokRyan2019}. Furthermore, the role of low-range conservation laws in diagnosing quantum integrability has recently been clarified substantially. The Reshetikhin condition, obtained as the first nontrivial compatibility condition in the expansion of a regular difference-form Yang--Baxter equation, has long been used as a necessary test for conventional Yang--Baxter integrability \cite{GrabowskiMathieu1995}. Recent results have shown that its significance is considerably stronger. Zhang proved that conservation of the canonical three-local charge is sufficient to generate higher local conserved quantities \cite{Zhang2025}, while Hokkyo established, in a rigorous infinite-chain setting, that the Reshetikhin condition implies an infinite hierarchy of mutually commuting local charges \cite{Hokkyo2026}. For isotropic nearest-neighbour spin chains, Shiraishi and Yamaguchi obtained a sharp dichotomy between complete integrability and the absence of nontrivial local conserved quantities \cite{ShiraishiYamaguchi2025}. Related classification results establish an analogous all-or-nothing structure for broad classes of symmetric spin-$\frac{1}{2}$ chains \cite{YamaguchiChibaShiraishi2024a,YamaguchiChibaShiraishi2024b}. Very recent work further argues that, under broad standard assumptions, the Reshetikhin condition is sufficient for the reconstruction of a regular analytic difference-form Yang--Baxter matrix \cite{Sanatani2026}.

Complementary progress has been made on rigorous demonstrations of non-integrability through the absence of finite-range conserved quantities. This program began with model-specific results for the $XYZ$ chain in a magnetic field and has since been extended to mixed-field Ising chains, next-nearest-neighbor models and general classes of symmetric spin chains \cite{Shiraishi2019,Chiba2024,ShiraishiNNN2024,HokkyoTest2025}. More recently, Fan \emph{et al.} proved the absence of local conserved charges in the Fredkin spin chain and several of its truncated versions, illustrating that exact ground-state solvability and unusual symmetry structures need not imply quantum integrability \cite{FanEtAlFredkin2025}. These developments make violations of the Reshetikhin conditions particularly significant: they give a concrete obstruction to the conventional regular Yang--Baxter mechanism. In this work we probe another avenue to generate conserved quantities for a Hamiltonian, namely the {\it scalene Yang-Baxter equation}. This  construction produces an exact cross-commuting transfer matrix and a nontrivial conserved polynomial algebra, thereby separating the symmetry-generating content of a Yang--Baxter-type relation from conventional transfer-matrix integrability.

The scalene Yang--Baxter equation replaces the three copies of a single $R$-matrix by three generally distinct operators $A,B,C$,
\begin{equation}
 A_{12}B_{13}C_{23}=C_{23}B_{13}A_{12}.
 \label{eq:scalene}
\end{equation}
Such relations were first studied in low dimensional Hilbert spaces by J. Hietarinta and C. Viallet \cite{HietarintaViallet2022}, where they discuss an eight-vertex rational solution that can be rotated to the $R$-matrix of the $XYZ$ spin chain. Apart from this the only work where the term scalene YBE was used can be found in a recent work \cite{konstantinou2026scalene}, that discusses certain set-theoretical solutions of these equations. The braid structure of the indices in these equations allows a suitable adaptation of the quantum inverse scattering method (QISM) to the scalene case as well. In particular, the meaning of the scalene transfer-matrices are significantly different from the standard equilateral YBE case. Rather than forcing one family to commute with itself, an invertible intertwiner $A$ naturally forces the commutativity of the transfer matrices $\tau_B$ and $\tau_C$ built from the distinct local operators $B$ and $C$ respectively. This raises a basic question: can the scalene equation detect conserved structures even when the regular $B$-operator does not define an ordinary Yang--Baxter integrable family?

We answer this question affirmatively with one explicit example. We find a regular $B$-matrix leading to a non-Hermitian, locally non-diagonalizable, Hamiltonian density that fails both ordinary regular-Yang--Baxter obstruction tests, the Reshetikhin conditions.  Its Pauli representation gives no transparent indication of the symmetry obtained from the scalene YBE. Instead, this symmetry is constructed using the $C$-transfer matrix, which in this case, evaluates exactly to a polynomial in the staggered raising operator $J_+$. Thus this shows that the scalene relation constructs a nontrivial operator that generates the commutant of the Hamiltonian.

It is important to highlight the distinction between symmetry and full integrability at this point. The powers generated by $\tau_C$ are algebraically dependent and in this case also trivially commute with each other. On the other hand $\tau_B(b)$ is not self-commuting for generic $b$. Our conclusion is therefore deliberately narrower than conventional Yang--Baxter integrability: scalene triples can generate hidden symmetries outside the ordinary regular $R$-matrix framework. This is a useful intermediate step in the search for scalene triples whose $C$-transfer matrices generate extensive sets of independent conserved quantities for the operators generated by the $B$-transfer matrix. Such examples would present genuine quantum integrable systems that violate both Reshetikhin conditions and thus go beyond the Yang-Baxter paradigm for integrability.

The paper is organized as follows. In Sec. \ref{sec:scalene-triple-nh-Hamiltonian} we introduce the exact non-braided scalene triple and extract the non-hermitian Hamiltonian from the regular $B$-matrix. Sec. \ref{sec:resh-violation} establishes the obstruction to an ordinary regular Yang--Baxter realization through the difference-form and non-difference-form Reshetikhin tests.  In Sec. \ref{sec:nilpotent-symmetry} we derive the cross-commutativity of the two transfer matrices and show that the tower of conserved charges for the non-hermitian Hamiltonian is generated by even powers of the staggered raising operator. This is done by evaluating the $C$-transfer matrix for arbitrary chain length and identifying the resulting staggered nilpotent hierarchy.  We conclude in Sec. \ref{sec:conclusion} with a discussion of the scope of the construction and its relation to conventional integrability. An appendix \ref{app:vectorization} is included to provide some missing details.

\section{The scalene triple and the non-hermitian Hamiltonian}
\label{sec:scalene-triple-nh-Hamiltonian}
For $b,x,y\in\mathbb C$, the triple,
\begin{align}
 A(x,y) =
\begin{pmatrix}
1&0&0&x\\ 0&-1&y&0\\ 0&0&1&0\\ 0&0&0&-1
\end{pmatrix}~&;~B(b)=\begin{pmatrix}
1&b&b&b\\ 0&b&1+b&0\\ 0&1+b&b&0\\ 0&0&0&1
\end{pmatrix}, 
\nonumber \\
C(x,y)&=\begin{pmatrix}
1&0&0&x\\ 0&1&0&0\\ 0&y&-1&0\\ 0&0&0&-1
\end{pmatrix},
\label{eq:ABC}
\end{align}
satisfies the non-braided scalene relation 
\begin{equation}
A_{12}(x,y)B_{13}(b)C_{23}(x,y)
=C_{23}(x,y)B_{13}(b)A_{12}(x,y),
\label{eq:exactSYBE}
\end{equation}
identically for arbitrary $b,x,y$. Due to this property the parameters $b,x,y$ can be seen as arbitrary functions of some spectral parameter. This avoids the need for an explicit {\it Baxterization} \cite{Jones1990} of the scalene triple, making the above solution ready to be used in the QISM framework.  

Certain properties of this triple immediately follow. We first note that $A$ and $C$ are invertible for arbitrary $x$ and $y$, while $B$ is invertible away from $b=-\frac{1}{2}$. More precisely,
\begin{equation}
\det A=1,\qquad \det B=-(1+2b),\qquad \det C=1,
\end{equation}
and
\begin{equation}
A^2=\mathbf1,\qquad C^2=\mathbf1.
\end{equation}
Thus we can use either $A$ or $C$ as an intertwiner in this triple. In this work we will let $A$ play this role.

The middle operator, $B$ is regular at $b=0$,
\begin{equation}
B(0)=P,\qquad B(b)=P+b~U,\qquad
U=\begin{pmatrix}
0&1&1&1\\0&1&1&0\\0&1&1&0\\0&0&0&0
\end{pmatrix},
\label{eq:U}
\end{equation}
with $P$ being the standard $4\times 4$ permutation operator.
Furthermore, as $PU=UP=U$, the two common regularity conventions give the same Hamiltonian density,
\begin{equation}
h=P\,\partial_bB(b)|_{b=0}=\partial_bB(b)|_{b=0}P= b'~U.
\label{eq:h-density}
\end{equation}
This  matrix is non-diagonalizable and not normal. Its minimal polynomial is $\lambda^2(\lambda-2)$, with the Jordan canonical form consisting of three blocks: one block of size 2 corresponding to $\lambda = 0$, one block of size 1 corresponding to $\lambda = 0$, and one block of size 1 corresponding to $\lambda = 2$.
Expanding \cref{eq:h-density} in the Pauli basis gives
\begin{align}
h={}&b'\left[\frac12\mathbf1\otimes\mathbf1-\frac12Z\otimes Z
+\frac14(\mathbf1\otimes X+X\otimes\mathbf1+X\otimes Z+Z\otimes X)\right.\nonumber\\
&\left.+\frac34X\otimes X+\frac14Y\otimes Y\right.\nonumber\\
&\left.+\frac{\ii}{4}(\mathbf1\otimes Y+Y\otimes\mathbf1+X\otimes Y+Y\otimes X+Y\otimes Z+Z\otimes Y)\right].
\label{eq:pauli-density}
\end{align}
Then the Hamiltonian on a periodic chain of length $L$
\begin{equation}
H=\sum_{j=1}^{L}h_{j,j+1},\qquad L+1\equiv1,
\end{equation}
becomes 
\begin{align}
H={}&b'\sum_j\bigg[\frac34X_jX_{j+1}+\frac14Y_jY_{j+1}-\frac12Z_jZ_{j+1}
+\frac14(X_jZ_{j+1}+Z_jX_{j+1})\nonumber\\
&\hspace{10mm}+\frac{\ii}{4}(X_jY_{j+1}+Y_jX_{j+1}+Y_jZ_{j+1}+Z_jY_{j+1})
+\frac12X_j+\frac{\ii}{2}Y_j\bigg] \nonumber\\[1mm]
=&\frac{b'}{2}\sum_j \left[X_jX_{j+1}+ Y_jY_{j+1}-Z_jZ_{j+1}+2~\sigma^+_j\sigma^+_{j+1}+\sigma^+_jZ_{j+1}+Z_j \sigma^+_{j+1}+\sigma^+_j+\sigma^+_{j+1}\right],
\label{eq:Hpauli}
\end{align}
up to an additive constant $b'L/2$. The model is manifestly non-Hermitian. It can be seen as a non-hermitian deformation of the $XXZ$ chain\footnote{This is beyond the known integrable deformations of the $XXZ$ spin chains \cite{beisert2013integrable}. It is also not a known twist deformation of the $XXX$ spin chains \cite{kulish2009twist} or part of known integrable non-hermitian spin chains \cite{Maity_2025}.}. It is not a conventional Jordan--Wigner quadratic chain as besides a $ZZ$ interaction, it contains parity-changing, mixed $XZ$ and $YZ$ terms. A simple inspection of the Hamiltonian \cref{eq:Hpauli} suggests that this model is devoid of any obvious symmetry. The scalene relation will be used to disprove this statement.

\section{Violation of Reshetikhin conditions}
\label{sec:resh-violation}
We now show that this density violates both the difference form and the non-difference form of the Reshetikhin conditions. This shows that this density does not arise from a standard regular Yang-Baxter solution or $R$-matrix. The claims concern differentiable homogeneous regular $4\times4$ ordinary $R$-matrices and are invariant under the usual addition of scalar and telescopic density terms.

\subsection{Difference-form condition}
\label{subsec:diff-conndition}
For a difference-form regular $R$-matrix, the Hamiltonian density must satisfy
\begin{equation}
\mathcal D:=[h_{12}+h_{23},[h_{12},h_{23}]]=K_{23}-K_{12}
\label{eq:staticR}
\end{equation}
for some two-site operator $K$ \cite{Reshetikhin1983,GrabowskiMathieu1995}. We can now explicitly construct the double commutator $\mathcal{D}$ for the Hamiltonian density given in \cref{eq:pauli-density}, we obtain an expression of the form
\begin{eqnarray}
\mathcal{D} &=& 2\sigma^+_2 \sigma^-_3 +(-1+2\sigma^-_2)\sigma^+_3-2 \sigma^+_2\sigma^+_3 -2\sigma^-_1\sigma^+_2(1 +2\sigma^+_3) -Z_1 \sigma^+_2  -2 Z_1\sigma^+_2\sigma^+_3 \nonumber\\[0.5mm]
&+& Z_2\sigma^+_3 + Z_1 Z_2 +2 Z_1 Z_2\sigma^+_3 +(\sigma^+_2 -Z_2) Z_3 +\sigma^+_1 \left(1-2\sigma^-_2 +2\sigma^+_2 +4\sigma^+_2\sigma^-_3 \right.\nonumber\\[0.5mm] &+&\left. 2\sigma^+_2 Z_3 -Z_2 -2 Z_2 Z_3 \right).
\end{eqnarray}
The operator $\mathcal{D}$ contains six genuine three-site operator terms containing indices 1, 2 and 3. Therefore, it cannot be reduced to a form involving only two-site operators, such as $\mathcal{D} =K_{23}-K_{12}$. Consequently, we expect no $K$-matrix satisfying the above relation to exist. 

The naive argument presented above provides an intuitive justification for the non-existence of the two-site $K$ operator. It can be made more rigorous by proving the following proposition.
\begin{proposition}[Non-existence of a two-site operator $K$]
\label{prop:non-K}
For the Hamiltonian density $h$ defined in \cref{eq:pauli-density} and the corresponding three-site operator $\mathcal{D} =[h_{12}+h_{23},[h_{12},h_{23}]]$, there is no $4 \times 4$ two-site operator $K$, with sixteen independent entries, satisfying 
\begin{equation*}
    \mathcal{D} =K_{23}-K_{12}.
\end{equation*}
\end{proposition}
\begin{proof}
We regard all sixteen entries of $K$ as independent unknowns and vectorize \cref{eq:staticR}.  Exact rational row reduction gives
\begin{equation*}
\operatorname{rank}\mathcal M_{\rm d}=15,\qquad
\operatorname{rank}(\mathcal M_{\rm d}\mid\operatorname{vec}\mathcal D)=16.
\end{equation*}
By using the theorem due to Rouch\'e--Capelli \cref{thm:rouche-capelli}, we conclude that no $K$-matrix exists, and the difference-form Reshetikhin condition fails. This Theorem and the associated vectorization procedure is described in detail in Appendix~\ref{app:vectorization}.
\end{proof}

\subsection{General non-difference condition}
\label{subsec:ndiff-condition}
\begin{proposition}[Failure of the necessary integrability condition]
\label{prop:failure-integrability}
Let
\begin{eqnarray*}
Q_2=\sum_j h_{j,j+1}~, \qquad
S=\sum_j [h_{j,j+1},h_{j+1,j+2}]~, \qquad
G(g)=\sum_j g_{j,j+1}~,
\end{eqnarray*}
on a periodic spin chain of length $L$, with indices understood modulo $L$.
Assume that for the Hamiltonian density under consideration,
\begin{eqnarray*}
g=\left.\partial_u h(u)\right|_{u=u_0} \propto h.
\end{eqnarray*}
Then
\begin{eqnarray*}
[Q_2,G(g)]=0.
\end{eqnarray*}
For the Hamiltonian density considered here, however, $[Q_2,S]\neq 0$ for generic periodic chain length $L$. Hence the necessary condition
\begin{eqnarray*}
[Q_2,G(g)]=[Q_2,S]
\label{eq:necessary-condition}
\end{eqnarray*}
cannot be satisfied.
\end{proposition}
\begin{proof}
For a regular non-difference $R(u,v)$, the local density $h(u)$ generally depends explicitly on the base spectral parameter. In the present case, however, the dependence is particularly simple
\begin{equation*}
    g=\left.\partial_b h(b')\right|_{b=0}
=\partial_b\left(b'~U\right)
=b'' ~U \propto h,
\end{equation*}
So $g$ is merely proportional to the Hamiltonian density $h$. This is readily observed from the explicit expression in \cref{eq:h-density}. As a result $ G(g) \propto Q_2$, and therefore $[Q_2, G(g)]=0$ holds trivially. 
On the other hand, evaluating the commutator $\left[Q_2 , S\right]$ explicitly, we find that
    \begin{eqnarray*}
     [Q_2,S] = \sum_j \left[h_{j,j+1}+h_{j+1,j+2},[h_{j,j+1},h_{j+1,j+2}]\right]= \sum_j \mathcal{D}_{j,j+1,j+2}.
\end{eqnarray*}
Here the site index $j$ is understood modulo $L$, corresponding to periodic boundary conditions. In particular, local density of $[Q_2,S]$ is precisely the operator $\mathcal{D} =[h_{12}+h_{23},[h_{12},h_{23}]]$. On a periodic chain, a translationally invariant sum of three-site operators $\sum_j \mathcal{D}_{j,j+1,j+2}$ can vanish identically only if the local operator $\mathcal{D}$ can be written as a telescoping difference of a two-site operator. That is, there must exist a two-site operator $K$ such that $\mathcal{D} =K_{23} -K_{12}$. Indeed, if such a $K$ exists, then 
\begin{eqnarray*} 
\sum_j \mathcal{D}_{j,j+1,j+2} =\sum_j \left( K_{j+1,j+2}-K_{j,j+1} \right), 
\end{eqnarray*} 
and the sum telescopes to zero on a periodic chain, 
\begin{eqnarray*}
    \sum_j \left( K_{j+1,j+2}-K_{j,j+1} \right) =0.
\end{eqnarray*}
We have already shown in \cref{prop:non-K} that no two-site operator $K$ satisfies this relation for the Hamiltonian density under consideration. Therefore, the three-site contributions cannot cancel through telescoping \footnote{This is the one-dimensional telescoping, or local coboundary, property of nearest-neighbor spin chains.}, and we conclude that
\begin{eqnarray*}
[Q_2,S]\neq 0.
\end{eqnarray*}
\end{proof}


In addition to the above tests for the Reshetikhin conditions we can also directly see that the non-hermitian Hamiltonian density does not arise from a regular $R$-matrix of either the difference or non-difference form. This direct obstruction displayed within the $B$-family itself.  Substituting the $B$-matrix into the ordinary non-difference form YBE
\begin{equation}
B_{12}(p)B_{13}(q)B_{23}(r)=B_{23}(r)B_{13}(q)B_{12}(p)
\end{equation}
reduces it to the additive form, $q=p+r$ only when $pr=0$.  This then corresponds to the degenerate branches $(p,q,r)=(0,r,r)$ and $(p,p,0)$, which are trivial solutions in the context of integrability. This thus implies that there is no genuine spectral composition with $pr\neq 0$.

\section{Cross-commutativity and the graded nilpotent hierarchy}
\label{sec:nilpotent-symmetry}
Let
\begin{equation}\label{eq:BC-monodromy}
T_a^{(B)}(b)=B_{aL}(b)\cdots B_{a1}(b),\qquad
T_a^{(C)}(x,y)=C_{aL}(x,y)\cdots C_{a1}(x,y)
\end{equation}
be homogeneous monodromy matrices, and define the corresponding transfer matrices $\tau_B=\tr_aT_B$ and $\tau_C=\tr_aT_C$, obtained after tracing over the auxiliary space indexed by $a$. With $A$ as the intertwiner it is easy to see that  
\begin{equation}\label{eq:ATbTc}
    A_{\mu\nu}(x,y)T_\mu^{(B)}(b)T_\nu^{(C)}(x,y)=T_\nu^{(C)}(x,y)T_\mu^{(B)}(b)A_{\mu\nu}(x,y),
\end{equation}
is satisfied by repeated use of the scalene relation. This is the analog of the $RTT$ relation in the equilateral case. We then have the following result for the cross-commutativity of the two transfer matrices.

\begin{proposition}[Scalene cross-commutativity]
If the local operators satisfy \cref{eq:exactSYBE} and the corresponding auxiliary intertwiner is invertible, then, with compatible homogeneous periodic monodromy conventions,
\begin{equation}
[\tau_B(b),\tau_C(x,y)]=0.
\label{eq:cross}
\end{equation}
Consequently every coefficient in any regular parameter expansion of $\tau_C$ commutes with $\tau_B(b)$ and with every charge obtained by expanding $\tau_B$ at its regular point.
\end{proposition}

\begin{proof}
The first statement of cross-commutativity, \cref{eq:cross} follows from multiplying both sides of \cref{eq:ATbTc} by $A^{-1}$ and tracing out the auxiliary indices $\mu$ and $\nu$, after using the cyclicity of the partial trace. 

The second statement is obtained by Taylor expanding $\tau_B$ about the regular point $b=0$, and $\tau_C$ about some fixed point $(x_0,y_0)$. For simplicity we choose the latter to be $(0,0)$ as well. We find that by substituting the corresponding Taylor expansions,
\begin{eqnarray}
    \tau_B & = & \sum\limits_{m=0}^\infty~\frac{b^m}{m!}\frac{\partial^m\tau_B}{\partial b^m}\Bigg|_{b=0} \equiv \sum\limits_{m=0}^\infty~b^mI_{B;m} \nonumber \\
    \tau_C & = & \sum\limits_{j=0}^\infty\sum\limits_{k=0}^\infty~\frac{x^jy^k}{j!k!}\frac{\partial^{j+k}\tau_C}{\partial x^j\partial y^k}\Bigg|_{x=0, y=0} \equiv \sum\limits_{j=0}^\infty\sum\limits_{k=0}^\infty~x^jy^k I_{C;j,k} \nonumber
\end{eqnarray}
into \cref{eq:cross} yields 
\begin{equation}\label{eq:cross-tower}
    \left[I_{B;m}, I_{C;j,k}\right] = 0~;~\forall~m,j,k.
\end{equation}

\end{proof}

It should be noted that \cref{eq:cross-tower} establishes cross-commutativity, and not self-commutativity of either the $B$ or $C$ charges. For instance this can already be seen for $L=2$, where one finds
\begin{equation}
[\tau_B(p),\tau_B(q)]=4pq(p-q)\,\Omega,
\label{eq:Bnoncomm}
\end{equation}
with 
\begin{equation}
    \Omega_{j,j+1}= Z_j \sigma^+_{j+1}+ \sigma^+_j Z_{j+1} - \sigma^+_j\sigma^+_{j+1}
\end{equation}

Thus the $B$-transfer matrices do not form a conventional commuting family for generic $p,q$. This feature further reaffirms the fact that the non-hermitian Hamiltonian \cref{eq:Hpauli} cannot be obtained from a regular $R$-matrix solving the equilateral YBE.

The above property changes for the $C$ transfer matrix, $\tau_C$.
Let $\sigma^{\pm}=\frac{X\pm\mathrm{i}Y}{2}$ be the standard raising(lowering) operator. Then we will show that the staggered global operator
\begin{equation}
J_+=\sum_{j=1}^{L}(-1)^j\sigma_j^+,
\label{eq:Jplus}
\end{equation}
is a symmetry of \cref{eq:Hpauli} and that its even powers can be extracted from the Taylor expansion of $\tau_C$.
In auxiliary-space block form the local $C$-operator is
\begin{equation}
C_{aj}(x,y)=
\begin{pmatrix}
\mathbf1_j&x\sigma_j^+\\ y\sigma_j^+&-\mathbf1_j
\end{pmatrix}_{a} = Z_a\mathbb{1}_j + y\sigma^-_a\sigma^+_j + x\sigma^+_a\sigma^+_j.
\label{eq:Cblock}
\end{equation}
The presence of $\sigma^{\pm}$ on the auxiliary space significantly simplifies the computation of $\tau_C$ even though $C$ is non-regular. This is shown in the following proposition.

\begin{proposition}
For a periodic chain of length $L$,
\begin{equation}
\tau_C(x,y)=
\begin{cases}
0,&L\ \text{odd},\\[1mm]
2\cos\!\left(\sqrt{xy}\,J_+\right),&L\ \text{even}.
\end{cases}
\label{eq:tCexact}
\end{equation}
The cosine terminates because $J_+^{L+1}=0$.
\end{proposition}

\begin{proof}
For the expansion shown in \cref{eq:Cblock}, the only term that contributes to the trace on the auxiliary space, from the monodromy matrix for C, \cref{eq:BC-monodromy} is the $2\times 2$ identity. Thus the computation of $\tau_C$ reduces to collecting the terms that produces this identity matrix on the auxiliary space, from the product of $L$, $C$ matrices, \cref{eq:Cblock}. Such terms can be obtained by counting the number of $Z_a$ matrices in a string of alternating $\sigma^{\pm}_a$'s. Then the traceful terms are those that contain an even number of $Z_a$'s and strings of the type $\sigma^{\pm}_a\cdots\sigma^{\mp}_a$ or $\sigma^{\mp}_a\cdots\sigma^{\pm}_a$. An immediate consequence of these requirements is that $\tau_C$ is identically 0 for odd $L$, as then we will necessarily violate one of the stated conditions.  As for the even $L$ case, we will identify $\tau_C$ by considering a few examples. For $L=2$ we find that 
\begin{eqnarray}
    \tr_a(C_{a2}C_{a1}) & = & \tr_a\left(\left[Z_a\mathbb{1}_2 + y\sigma^-_a\sigma^+_2 + x\sigma^+_a\sigma^+_2 \right]\left[Z_a\mathbb{1}_1 + y\sigma^-_a\sigma^+_1 + x\sigma^+_a\sigma^+_1 \right]\right) \nonumber \\
    & = & 2\mathbb{1} + 2xy~\sigma^+_1\sigma^+_2.
\end{eqnarray}
The traceful terms are obtained from a term containing two $Z_a$'s, the $\mathbb{1}$ term, and no $Z_a$, the $\sigma^+_1\sigma^-_2$ term. When $L=4$ the traceful terms contain either four $Z_a$'s, two $Z_a^2$ or no $Z_a$'s. The first of these just lead to $\mathbb{1}$ on the physical space. The second set is inserted in bilinear strings of $\sigma^+_j\sigma^-_k$ and the third in a quadrilinear string of alternating $\sigma^{\pm}$'s. The net result is,  
\begin{eqnarray}
    & \tau_C(x,y) =  \tr_a\left(C_{04}C_{03}C_{02}C_{01}\right)     & \nonumber \\
    &  = \tr_a\left(\left[Z_0\mathbb{1}_4 + y \sigma^-_0\sigma^+_4 + x \sigma_0^+ \sigma^+_4 \right]\left[Z_0\mathbb{1}_3 + y \sigma^-_0\sigma^+_3 + x \sigma_0^+ \sigma^+_3 \right]\right.    & \nonumber \\
    &\times \left.\left[Z_0\mathbb{1}_2 + y \sigma^-_0\sigma^+_2 + x \sigma_0^+ \sigma^+_2 \right] \left[Z_0\mathbb{1}_1 + y \sigma^-_0\sigma^+_1 + x \sigma_0^+ \sigma^+_1 \right]\right)       & \nonumber \\
    & = 2\mathbb{1} + 2\frac{xy}{1!}\left( \sigma_1^+\sigma_2^+ + \sigma_2^+\sigma_3^+ + \sigma_3^+\sigma_4^+ + \sigma_4^+\sigma_1^+ - \sigma_1^+\sigma_3^+ - \sigma_2^+\sigma_4^+ \right) &  \nonumber \\
    & + 2\frac{x^2y^2}{2!} \sigma_1^+\sigma_2^+\sigma_3^+\sigma_4^+  &
\end{eqnarray}
It is easily observed that the coefficient of the $(xy)^{2k}$ is the same even power of the staggered nilpotent operator \cref{eq:Jplus}, and thus the claimed result follows for the even case.

\end{proof}

Thus for even $L$ we have,
\begin{equation}
\tau_C(x,y)=2\sum_{n=0}^{\lfloor L/2\rfloor}
\frac{(-xy)^n}{(2n)!}J_+^{2n}.
\label{eq:tCseries}
\end{equation}
The scalene transfer matrix therefore generates the finite conserved hierarchy
\begin{equation}
\mathbf1,\ J_+^2,\ J_+^4,\ldots,\ J_+^{2\lfloor L/2\rfloor}.
\end{equation}
It is a tower in the transfer-matrix sense, but not a tower of algebraically independent charges: every coefficient lies in the even subalgebra of
\begin{equation}
\mathbb C[J_+]/(J_+^{L+1}).
\end{equation}
At this point it is worth noting that the problem of algebraic independence of the conserved charges plagues even standard integrable spin chains such as the $XXX$ model. In fact, algebraic independence of every transfer-matrix coefficient is not itself guaranteed by the conventional Yang–-Baxter construction. On any finite chain it is, strictly speaking, impossible because the charges are finite-dimensional matrices and hence they should satisfy polynomial identities. For standard models such as the $XXX$ chain, the logarithmic derivatives of the fundamental transfer matrix are instead regarded as nonredundant in the thermodynamic local-operator sense: successive charges contain connected densities of increasing range \cite{GrabowskiMathieu1995-2}. Even this conventional local family is not complete in the strongest thermodynamic sense, since further quasilocal charges linearly independent of it are known \cite{IlievskiMedenjakProsen2015,IlievskiStringCharge2016}. In contrast to this the hierarchy for the non-hermitian Hamiltonian \cref{eq:Hpauli}, exhibits a stronger, explicitly identifiable dependence, since every coefficient is a power of $J_+^2$.

It should therefore be interpreted as a transfer-matrix realization of the conserved polynomial algebra generated by a hidden nilpotent symmetry, rather than as an algebraically independent family of local integrals of motion.

The fact that the even powers of the staggered nilpotent operator \cref{eq:Jplus} are conserved motivates us to check if \cref{eq:Jplus} is conserved by itself. In fact it turns out that this stronger linear symmetry can be checked locally. We find that 
\begin{equation}
[h_{j,j+1},\sigma_j^+-\sigma_{j+1}^+]=0.
\label{eq:localsym}
\end{equation}
Multiplying alternating local identities by signs and summing on an even periodic chain gives
\begin{equation}
[H,J_+]=0.
\label{eq:HJsym}
\end{equation}
This direct verification is consistent with \cref{eq:cross,eq:tCexact}, but the important point is how \cref{eq:Jplus} was found: it emerges from an exact auxiliary member of the scalene triple, rather than from an ansatz for global symmetries of the complicated Pauli Hamiltonian \cref{eq:Hpauli}.

Because $J_+$ is nilpotent, it does not yield the usual decomposition into eigenspaces of a semisimple symmetry.  Instead it defines an invariant filtration
\begin{equation}
0\subset\ker J_+\subset\ker J_+^2\subset\cdots\subset\ker J_+^{L+1}=\cH.
\label{eq:filtration}
\end{equation}
Furthermore, \cref{eq:HJsym} implies that every subspace in \cref{eq:filtration} is preserved by $H$.  A basis adapted to this filtration therefore block-upper-triangularizes the Hamiltonian and may be useful for analyzing its Jordan structure.

\section{Discussion and outlook}
\label{sec:conclusion}
In this work we have considered an example of a non-hermitian deformation of the $XXZ$ model, \cref{eq:Hpauli}, constructed from a non-braided scalene triple that satisfies the scalene YBE \cref{eq:exactSYBE}. This model has three key features:
\begin{enumerate}
    \item the Hamiltonian density violates both the ordinary regular-YBE Reshetikhin tests, namely the ones derived from $R$-matrices of the difference and non-difference form. This implies that this Hamiltonian cannot be derived from a regular $R$-matrix satisfying the standard or equilateral YBE.
    \item the above point is further reinforced by the absence of a self-commuting $B$-transfer family, from which this Hamiltonian is derived. This implies that there is no operator that acts as an intertwiner between the $B$-matrices.
    \item nevertheless, cross commutativity of the $B$ and $C$-transfer matrices, \cref{eq:cross} helped us derive the commutant of the Hamiltonian from the algebra generated using $\tau_C$. In the example considered this turned out to be generated by a non-trivial staggered nilpotent symmetry \cref{eq:Jplus}.
\end{enumerate}

Apart from this it is important to note that the $A$-matrix acts as an intertwiner between a regular $B$-matrix and a non-regular $C$-matrix, which is impossible in the standard or equilateral YBE setting. 
Thus this example marks a precise boundary between conventional transfer-matrix integrability, obtained from the standard YBE, and the broader notion of integrability suggested by the scalene YBE.

Thus the example considered in this work is not integrable in the standard sense. The coefficients of the $C$-transfer matrix are functions of one nilpotent generator, rather than an extensive set of algebraically independent local charges, while \cref{eq:Bnoncomm} rules out the usual commuting $B$-family. So through this example we have established that  \begin{equation*}
\text{scalene YBE}\quad\Longrightarrow\quad
\text{cross-commuting transfer matrices}\quad\Longrightarrow\quad
\text{hidden conserved symmetry}.
\end{equation*}
Clearly a more desirable scalene triple would be one that would result in a system with a more non-trivial tower of conserved charges similar to the ones seen in the case of standard integrable models. However, constructing scalene triples seems to be much harder than the standard YBE. Among the many solutions we studied, another one is worth mentioning due to the models obtained from that triple. Consider the solution
\begin{eqnarray}
  &  A=
\begin{pmatrix}
0 & 0 & 0 & a_1 \\[2pt]
0 & i\sqrt{a_1}\sqrt{a_2} & 0 & 0 \\[2pt]
0 & 0 & -i\sqrt{a_1}\sqrt{a_2} & 0 \\[2pt]
a_2 & 0 & 0 & 0
\end{pmatrix}~;~ B=
\begin{pmatrix}
1+b_1 & 0 & 0 & b_5 \\[2pt]
0 & b_2 & 1+b_6 & 0 \\[2pt]
0 & 1+b_7 & b_3 & 0 \\[2pt]
b_8 & 0 & 0 & 1+b_4
\end{pmatrix} & \nonumber \\
& C
= \begin{pmatrix}
-\dfrac{i\sqrt{a_1}\,b_3}
        {\sqrt{a_2}(1+b_6)}
&
0
&
0
&
\dfrac{a_1(1+b_7)}
      {a_2(1+b_6)}
\\[10pt]
0
&
\dfrac{i\sqrt{a_1}(1+b_4)}
      {\sqrt{a_2}(1+b_6)}
&
-\dfrac{a_1b_8}
       {a_2(1+b_6)}
&
0
\\[10pt]
0
&
-\dfrac{b_5}{1+b_6}
&
-\dfrac{i\sqrt{a_1}(1+b_1)}
       {\sqrt{a_2}(1+b_6)}
&
0
\\[10pt]
1
&
0
&
0
&
\dfrac{i\sqrt{a_1}\,b_2}
      {\sqrt{a_2}(1+b_6)}
\end{pmatrix}.  & 
\end{eqnarray}
In this case we can use $C$ as the intertwiner between $A$ and $B$, establishing the cross-commutativity of their respective transfer matrices. As in the example studied in this work, here too $B$ is regular, while $A$ is not\footnote{$A$ is similar to the braided form of the $(1,4)$ class of Hietarinta's classification of $4\times 4$ Yang-Baxter solutions \cite{HIETARINTA-PLA,MSPK-Hietarinta}.}. The Hamiltonian density from the $B$-transfer matrix is a general eight-vertex example with arbitrary parameters $b_j$'s. Such a model is generally not expected to be integrable in the Yang-Baxter sense. Consistent with that we find that the $A$-transfer matrix only generates operators that are functions of the global parity and the translational symmetry of the periodic chain. 

We will leave the search for the more interesting scalene triple that will truly establish integrability beyond the standard Yang-Baxter framework for a future work.

As a final point, it would also be interesting to study the properties non-hermitian Hamiltonian constructed in this work. This is because non-Hermitian quantum spin chains arise naturally as effective descriptions of open and nonequilibrium many-body systems. In the quantum-trajectory formulation of a Lindblad evolution, evolution conditioned on the absence of detected quantum jumps is generated by a non-Hermitian effective Hamiltonian \cite{DalibardCastinMolmer1992,PlenioKnight1998,AshidaGongUeda2020}. Interacting non-Hermitian chains can therefore encode the conditional relaxation, transport and entanglement dynamics of monitored or dissipative quantum systems. Closely related non-Hermitian spin operators appear as stochastic generators of classical reaction--diffusion and exclusion processes, where their spectra determine relaxation rates and nonequilibrium steady-state properties \cite{AlcarazEtAl1994,Schutz2001}. These connections suggest that the Hamiltonian constructed here may provide an exactly structured starting point for investigating conditioned open-system dynamics and nonequilibrium relaxation in the presence of a hidden nilpotent symmetry. In particular, the exact conservation of the staggered nilpotent operator may constrain decay channels, produce invariant sectors of the conditioned dynamics, and influence the formation of exceptional points or nontrivial Jordan blocks. Determining a microscopic Lindblad or stochastic realization of the present Hamiltonian, and establishing how the nilpotent symmetry manifests itself after the recycling terms associated with quantum jumps are restored, constitute interesting directions for future work.




\section*{Acknowledgements}
PP thanks Jarmo Hietarinta for the introduction to the scalene Yang--Baxter relations. PP also thanks the organizers of the The 2nd ISNMP Conference, Bad Ems, Germany, where this work was initiated.

\appendix
\section{Rigorous proof for the violation of the first Reshetikhin condition}
\label{app:vectorization}
Let us begin by recalling the Rouch\'e--Capelli theorem \cite[p. 56, Thm. 2.38]{shafarevich2012linear}, which provides the rank criterion needed to determine whether the linear system for the unknown $K$ has a solution or not.
\begin{theorem}[Rouch\'e--Capelli]
\label{thm:rouche-capelli}
Let $M$ be an $m\times n$ matrix over a field $\mathbb{F}$, and let
$x \in \mathbb{F}^{n}$ and $d \in \mathbb{F}^{m}$. Denote by $(M\mid d)$ the augmented matrix obtained by appending $d$ as an additional column to $M$. Then the linear system
\begin{eqnarray*}
M x = d
\end{eqnarray*}
has a solution if and only if
\begin{eqnarray*}
\operatorname{rank}(M) = \operatorname{rank}(M\mid d).
\end{eqnarray*}
\end{theorem}
We will use the criterion set by this Theorem and complete the linear-algebraic justification for Proposition~\ref{prop:non-K}. The central idea is that the unknown operator $K$ enters linearly in \cref{eq:staticR}. Therefore, we can convert the constraint into an ordinary linear system for the entries of $K$. Let $E_{1},E_{2},\dots,E_{16}$ be a basis of the $16$ dimensional vector space of $4 \times 4$ matrices. In this way, one may represent the unknown two-site operator $K$ by
\begin{equation}
K = \sum_{a=1}^{16} k_{a} E_{a},
\end{equation}
with the coefficients $k_{1},\dots,k_{16}$ being the $16$ independent unknowns.
Therefore,
\begin{equation}
K_{23} - K_{12} = \sum_{a=1}^{16} k_{a}\left(
id \otimes E_{a}-E_{a}\otimes id
\right) = \sum_{a=1}^{16} k_{a} \mathcal{L}_{a},
\end{equation}
where $\mathcal{L}_a  = \left(id \otimes E_{a}-E_{a}\otimes id\right)$ for $a=1,2,\ldots, 16$. Vectorizing \cref{eq:staticR}, we get
\begin{equation}
\operatorname{vec}\mathcal{D}= \sum_{a=1}^{16}k_{a}\,
\operatorname{vec} \mathcal{L}_{a}.
\end{equation}
Equivalently,
\begin{equation}
\mathcal{M}_{\mathrm d}~\mathbf{k} = \operatorname{vec}\mathcal{D},
\end{equation}
where 
\begin{equation}
\mathcal{M}_{\mathrm d}:=
\begin{bmatrix}
\vert & \vert & & \vert \\
\operatorname{vec} \mathcal{L}_{1} & \operatorname{vec}\mathcal{L}_{2} & \cdots & \operatorname{vec} \mathcal{L}_{16} \\
\vert & \vert & & \vert
\end{bmatrix},
\qquad
\mathbf{k} :=
\begin{pmatrix}
k_{1} \\
k_{2} \\
\vdots \\
k_{16}
\end{pmatrix}.
\end{equation}
The exact symbolic computation gives
\begin{equation}
\operatorname{rank}\mathcal{M}_{\mathrm{d}}=15~,
\qquad
\operatorname{rank}\left(\mathcal{M}_{\mathrm{d}}\mid \operatorname{vec}\mathcal{D}\right) =16.
\end{equation} 
The matrix $\mathcal{M}_{\mathrm{d}}$ maps the $16$- dimensional space of $K$'s entries into the $64$- dimensional space. The rank of $\mathcal{M}_{\mathrm{d}}$ is $15$ means the image of this map is $15$-dimensional. So the dimension of the kernel is $1$. This one-dimensional kernel has a simple interpretation. If $K$ is proportional to the identity,
\begin{equation}
K = \alpha~ I,
\end{equation}
then
\begin{equation}
K_{23}-K_{12}=0.
\end{equation}
Thus the identity matrix always lies in the kernel. This explains why the rank is $15$ rather than $16$. Now consider the augmented matrix
\begin{equation}
\left(\mathcal{M}_{\mathrm{d}}\mid \operatorname{vec}\mathcal{D}\right),
\end{equation}
which is obtained by appending the target vector $\operatorname{vec} \mathcal{D}$ as an additional column. If $\operatorname{vec}\mathcal{D}$ could be written as $\mathcal{M}_{\mathrm{d}}~\mathbf{k}$ for some choice of $K$, then adding $\operatorname{vec}\mathcal{D}$ would not increase the rank. However, the computation shows that
\begin{equation}
\operatorname{rank}\left(\mathcal{M}_{\mathrm{d}}\mid \operatorname{vec}\mathcal{D}\right)=16.
\end{equation}
Thus $\operatorname{vec}\mathcal{D}$ is linearly independent of the columns of $\mathcal{M}_{\mathrm{d}}$. Geometrically, $\operatorname{vec}\mathcal{D}$ does not lie in the $15$-dimensional image of the map. We see that the ranks of $\mathcal{M}_{\mathrm{d}}$ and $(\mathcal{M}_{\mathrm{d}}\mid \operatorname{vec}\mathcal{D})$ are different. Therefore, no two-site operator $K$ exists satisfying the \cref{eq:staticR}.

\bibliographystyle{ieeetr}
\bibliography{refs}
\end{document}